\documentclass[10pt,journal,draftclsnofoot,onecolumn]{IEEEtran}
\usepackage{}

\usepackage{mathpazo}
\usepackage{mathcomp}
\usepackage{setspace}
\usepackage{amsmath}
\usepackage{amssymb}
\usepackage{mathrsfs}
\usepackage{amsthm}
\usepackage{multirow}
\usepackage{color}
\usepackage{longtable}
\usepackage{array}
\usepackage{url}
\usepackage{comment}
\usepackage{enumerate}
\usepackage{eucal}

\usepackage{algorithm}
\usepackage[noend]{algorithmic}
\usepackage{cite}
\usepackage[table]{xcolor}

\usepackage[top=1.05in, bottom=0.8in, left=0.79in, right=0.79in]{geometry}

\theoremstyle{plain}
\newtheorem{theorem}{Theorem}
\newtheorem{corollary}[theorem]{Corollary}
\newtheorem{lemma}[theorem]{Lemma}
\newtheorem{proposition}[theorem]{Proposition}

\theoremstyle{definition}
\newtheorem{definition}[theorem]{Definition}
\newtheorem{example}[theorem]{Example}

\newtheorem{construction}{Construction}

\newcommand{\F}{\mathbb{F}}

\DeclareMathAlphabet{\mathbfsl}{OT1}{ppl}{b}{it} 

\newcommand{\vr}{\mathbfsl{r}}
\newcommand{\vu}{\mathbfsl{u}}
\newcommand{\vv}{\mathbfsl{v}}
\newcommand{\vx}{\mathbfsl{x}}
\newcommand{\vy}{\mathbfsl{y}}

\newcommand{\vc}{\mathbfsl{c}}

\newcommand{\vA}{\mathbfsl{A}}
\newcommand{\vG}{\mathbfsl{G}}
\newcommand{\vI}{\mathbfsl{I}}

\newcommand{\mC}{\mathbb{C}}
\newcommand{\mD}{\mathbb{D}}
\newcommand{\mE}{\mathbb{E}}

\newcommand{\cB}{\mathcal{B}}
\newcommand{\cC}{\mathcal{C}}
\newcommand{\cD}{\mathcal{D}}
\newcommand{\cE}{\mathcal{E}}

\newcommand{\cP}{\mathcal{P}}

\newcommand{\xxx}{\mathbfsl{x}}
\newcommand{\yyy}{\mathbfsl{y}}

\newcommand{\bta}{{\pmb{\beta}}}
\newcommand{\bsg}{{\pmb{\sigma}}}

\newcommand{\bt}{{\pmb{\tau}}}

\newcommand{\supp}{{\rm supp}}
\newcommand{\wt}{{\rm wt}}

\newcommand{\zero}{{\mathbf 0}}

\newcommand{\floor}[1]{{\left\lfloor #1\right\rfloor}}

\title{Low-Power Cooling Codes\\ with Efficient Encoding and Decoding
\vspace{0.2cm}}

\author{
\IEEEauthorblockN{
Yeow Meng Chee\IEEEauthorrefmark{1},
Tuvi Etzion\IEEEauthorrefmark{2},
Han Mao Kiah\IEEEauthorrefmark{1},
Alexander Vardy \IEEEauthorrefmark{1}$^,$\IEEEauthorrefmark{3},
Hengjia Wei\IEEEauthorrefmark{1}}\\[1mm] \vspace{0.3cm}
\IEEEauthorblockA{\IEEEauthorrefmark{1}\footnotesize School of Physical and Mathematical Sciences, Nanyang Technological University, Singapore} \\ \vspace{0.2cm}
\IEEEauthorblockA{\IEEEauthorrefmark{2}\footnotesize Department of Computer Science, Technion, Israel} \\ \vspace{0.2cm}
\IEEEauthorblockA{\IEEEauthorrefmark{3}\footnotesize Dept. of Electrical and Computer Engineering and Dept. of Computer Science and Engineering, University of California San Diego, USA}
\thanks{Earlier results of this paper were presented at the 2018 Proceedings of the IEEE ISIT \cite{CEKVW18}.}
\vspace{-10mm}
}

\begin{document}
\date{}
\maketitle

\hspace*{-1pt}\begin{abstract}\boldmath
A class of low-power cooling (LPC) codes,
to control simultaneously both the peak temperature and the average power consumption of interconnects,
was introduced recently.
An $(n,t,w)$-LPC code is a coding scheme over $n$ wires that\newline
(A) avoids state transitions on the $t$ hottest wires (cooling), and\newline
(B) limits the number of transitions to $w$ in each transmission (low-power).

A few constructions for large LPC codes that have efficient encoding and decoding schemes, are given.
In particular, when $w$ is fixed, we construct LPC codes of size $(n/w)^{w-1}$ and
show that these LPC codes can be modified to correct errors efficiently.
We further present a construction for large LPC codes based on a mapping
from cooling codes to LPC codes. The efficiency of the encoding/decoding for the constructed LPC codes
depends on the efficiency of the decoding/encoding for the related cooling codes and the ones
for the mapping.
\end{abstract}

\section{Introduction}
Power and heat dissipation
have emerged as first-order design constraints for chips,
whether targeted for battery-powered devices or for high-end systems.
High temperatures have dramatic negative effects on interconnect performance.
Power-aware design alone is insufficient to
address the thermal challenges, since it does not directly target the
spatial and temporal behavior of the operating environment. For this
reason, thermally-aware approaches have emerged as one~of~the most
important domains of research in chip design today.
Numerous techniques have been proposed to reduce the overall power consumption of on-chip buses
(see \cite{Cheeetal.2017} which uses coding techniques and the references therein using non-coding techniques).
However, all the non-coding techniques do not directly address peak temperature minimization.

Recently, Chee et al. \cite{Cheeetal.2017} introduced several efficient coding schemes to directly control the peak temperature and the average power consumption.
Among others, {\it low-power cooling (LPC) codes} are of particular interest as
they control both the peak temperature and the average power consumption simultaneously.
Specifically, an $(n,t,w)$-LPC code is a coding scheme for communication over a bus consisting of $n$ wires,
if the scheme has the following two features:
\begin{enumerate}
\item[(A)]  every transmission does not cause state transitions on the $t$ hottest wires;
\item[(B)] the number of state transitions on all the wires is at most $w$ in every transmission.
\end{enumerate}

LPC codes have both features, while cooling codes control only the peak temperature.
\begin{definition}
For $n$ and $t$, an {\it $(n,t)$-cooling code} $\mC$ of size $M$ is defined as a collection
$\{\cC_1, \cC_2, \ldots, \cC_M\}$, where $\cC_1, \cC_2, \ldots, \cC_M$ are disjoint subsets
of $\{0,1\}^n$ satisfying the following property:
for any set ${S} \subseteq [n]$ of size $|S| = t$ and for $i\in[M]$,
there exists a vector $\vu \in \cC_i$ such that $\supp(\vu)\cap S = \varnothing$.
We refer to  $\cC_1,\cC_2, \ldots, \cC_M$ as {\it codesets} and the vectors  in them as {\it codewords}.
\end{definition}

Using {\em partial spreads}, Chee et al.~\cite{Cheeetal.2017} constructed LPC codes with efficient encoding and decoding schemes.
When $t\le 0.687n$ and $w\ge (n-t)/2$, these codes achieve optimal asymptotic rates.
However, when $w$ is small, i.e. low-power coding is used, the code rates are small and
Chee et al. proposed another construction based on {\em decomposition of the complete hypergraph} into perfect matchings.
While the construction results in LPC codes of large size, usually efficient encoding and decoding algorithms are not known.

In this work, we focus on this regime ($w$ small) and construct LPC codes with efficient encoding and decoding schemes.
Specifically, our contributions are as follows.

\begin{enumerate}[(I)]
\item We propose a method that takes a linear erasure code as input and constructs an LPC code.
Using this method, we then construct a family of LPC codes of size $(n/w)^{w-1}$
which attains the asymptotic upper bound $O(n^{w-1})$ when $w$ is fixed.
We also use this method to construct a class of LPC codes of size $(n/w)^{w-e-1}$ which is able to correct $e$ transmission errors.

\item {We propose efficient encoding/decoding schemes for the LPC codes of the given construction.
In particular, for the above family of LPC codes, we demonstrate encoding with
$O(n)$ multiplications over $\F_q$  and decoding with $O(w^3)$ multiplications over $\F_q$, where $q={n}/{w}$.}
Furthermore, the related class of LPC codes is able to correct $e$ errors with complexity $O(n^3)$.

\item A definition for a new family of low-power cooling codes, called \emph{constant-power cooling}
(CPC in short) codes, which have the same weight for all the
codewords. All our previous constructions can be applied to obtain such codes.
\item A recursive construction for a class of $(nq,tq,w)$-CPC codes (and also $(nq,tq,w)$-LPC codes) from $(n,t,w)$-CPC codes
(and a special type of $(n,t,w)$-LPC codes).

\item A construction for a class of $(n,t,w)$-LPC codes based on a mapping from $(m,t)$-cooling
codes. This mapping send all the binary words of length $m$ into a Hamming ball of radius $w$ in $\F_q^n$
such that each coordinate of the obtained words in the Hamming ball is dominated by a coordinate of
the binary words of length $m$. This property guarantees that the cooling property of the $(m,t)$-cooling code
is preserved in the low-power cooling code.
\end{enumerate}

Our main target in this paper are cooling codes, but these codes, or more precisely
their definition by codesets, might have other applications too. One such application
is in the design of WOM (Once Write Memory) codes which are very important in
coding for flash memories (see~\cite{CKVY17} and references therein). This application of codesets
into construction of WOM codes was written in detail in~\cite{CKVY17} and
is described in short as follows. In a WOM code one is trying to write binary information words
of length $k$ into a memory of length $n$, where the information is written only in
positions where there are \emph{zeroes}. The goal is to write as many rounds as possible
until there is no way to distinguish between some of the written words. Each information
word will be identified by a codeset, the codeword taken from the appropriate codeset
should have \emph{ones} on all positions where the memory has \emph{ones} and hence
the \emph{ones} in complement of the codeword should have empty intersection with the
\emph{ones} of the memory. As was mentioned, this application to WOM codes was considered in~\cite{CKVY17}.
We believe that other applications will arise in the future.

The rest of this paper is organized as follows. In Section~\ref{sec:upper} we present some
necessary definitions for our exposition, some of the known results, and new upper
bounds on the sizes of low-power cooling codes and constant-power cooling codes.
Finally, the known constructions are presented and a new one is suggested.
Section~\ref{sec:efficient} suggests a construction for CPC codes based on non-binary linear
codes in general and on MDS codes in particular. For these codes efficient encoding
and decoding algorithms are derived. We continue in Section~\ref{sec:ECC} and add error-correction capabilities
for such codes and provide efficient algorithms also in this case. The construction that was used
in Section~\ref{sec:ECC} is modified in Section~\ref{sec:recursion} to provide a recursive construction for
$(nq,tq,w)$-CPC codes (and related $(nq,tq,w)$-LPC codes) from $(n,t,w)$-CPC codes (and some special $(n,t,w)$-CPC codes).
While in Section~\ref{sec:efficient} the constructions are for $t\le n/w -1$, in Section~\ref{sec:recursion}
the construction is effective for larger $t$. In Section~\ref{sec:map}
a method to transfer an $(m,t)$-cooling code to an $(n,t,w)$-LPC code is given.
This method is based on a special injection from the set of all binary words of length $m$
into binary words of length $n$ and weight at most $w$. A product construction using
this method implies codes with efficient encoding and decoding algorithms.
We further analyse and compare between this construction and constructions in previous works.

\section{Upper Bounds and Known Results}
\label{sec:upper}

Given a positive integer $n$, the set $\{1, 2, \ldots , n\}$ is abbreviated as $[n]$.
The {\it Hamming weight} of a vector $\vx \in \F_q^n$, denoted $\wt(\vx)$, is the number of nonzero positions in $\vx$,
while the {\it support} of $\vx$ is defined as $\supp(\vx)\triangleq\{ i\in[n]: x_i\not=0\}$.

A $q$-ary {\it code $\cC$ of length $n$} is a subset of $\F_q^n$.
If $\cC$ is a subspace of $\F_q^n$, it is called a {\it linear code}.
An $[n,k,d]_q$ code is a linear code with dimension $k$ and minimum Hamming distance $d$.

\begin{definition}
For $n$, $t$ and $w$ with $t+w \leq n$, an {\it $(n,t,w)$-low-power cooling (LPC) code} $\mC$ of size $M$
is defined as a collection of codesets $\{\cC_1, \cC_2, \ldots, \cC_M\}$, where $\cC_1, \cC_2, \ldots, \cC_M$
are disjoint subsets of $\{\vu\in \{0,1\}^n: \wt(\vu) \leq w\}$ satisfying the following property:
for any set ${S} \subseteq [n]$ of size $|S| = t$ and for $i\in[M]$,
there exists a vector $\vu \in \cC_i$ such that $\supp(\vu)\cap S = \varnothing$.
\end{definition}

In this paper, we focus on a class of  $(n,t,w)$-LPC codes where
every transmission results in exactly $w$ state transitions.
We call such codes {\it $(n,t,w)$-constant-power cooling (CPC) codes}.
In particular, let $J(n,w)\triangleq \{\vu\in\{0,1\}^n: \wt(\vu) = w\}$.
Then an $(n,t,w)$-CPC code is an $(n,t,w)$-LPC code such that $\cC_i\subseteq J(n,w)$ for each $i\in[M]$.

\subsection{Set Systems}

For a finite set $X$ of size $n$, $2^X$ denotes the collection of all subsets of $X$, i.e.,
$2^X\triangleq\{A: A\subseteq X\}$.  A {\it set system} of order $n$ is a pair $(X, \cB)$,
where $X$ is a finite set of $n$ {\it points}, $\cB \subseteq 2^X$, and the elements of $\cB$ are called {\it blocks}.
Two set systems $(X, \cB_1)$ and $(X, \cB_2)$
with the same point set are called {\it disjoint} if $\cB_1 \cap \cB_2 =\varnothing$, i.e. they don't have
any block in common.

A {\it partial parallel class} of a set system $(X,\cB)$ is a collection of pairwise disjoint blocks.
If a partial parallel class partitions the point set $X$, it is called {\it parallel class}.
A set system $(X, \cB)$ is called {\it resolvable} if the block set $\cB$ can be partitioned into parallel classes.

There is a canonical one-to-one correspondence between the set of all codes of length $n$ and
the set of all set systems of order $n$: the coordinates of vectors in $\{0,1\}^n$ correspond
to the points in $[n]$, and each vector $\vu\in \{0,1\}^n$ corresponds to the block defined by $\supp(\vu)$.
Thus we may speak of the set system of a code or the code of a set system.
By abuse of notation we sometimes do not distinguish between the two different notations and this
can be readily observed in the text.

\subsection{Upper Bounds}

Given a $t$-subset $S$ and a vector $\vu \in \{0,1\}^n$, we shall say that $\vu$ {\it avoids} $S$
if $\supp(\vu)\cap S = \varnothing$. The following bounds on LPC codes and CPC codes are easily derived.

\begin{theorem}
\label{thm:upper_bounds}
Let $\mC$ be an $(n,t,w)$-LPC code of size $M$, then
\begin{align*}
M  \leq  \sum_{i=0}^w {n-t\choose i}.
\end{align*}
Furthermore, if $\mC$ is an $(n,t,w)$-CPC code, then
\begin{align*}
M  \leq  {n-t\choose w}.
\end{align*}
\end{theorem}
\begin{proof}
For any given $t$-subset $S$ of $[n]$, each codeset should have at least one codeword
which avoids $S$. The number of words with weight $i$ which avoid $S$ is ${n-t \choose i}$ and hence
there are no more than ${n-t\choose w}$ codesets in an $(n,t,w)$-CPC code and no more than
$\sum_{i=0}^w {n-t\choose i}$ codesets in an $(n,t,w)$-LPC code.
\end{proof}

Theorem~\ref{thm:upper_bounds} implies that both CPC codes and LPC codes share the same asymptotic upper bound $O(n^w)$
on the number of codewords. The upper bound of Theorem~\ref{thm:upper_bounds} can
be improved for some parameters. For this purpose, we need to define and to introduce some results on Tur\'an systems.

Let $n\geq k\geq r$, and let $X$ be a finite set with $n$ distinct elements. The set $\binom{X}{r}$ is the collection of all $r$-subsets of $X$.
A {\it Tur\'an $(n,k,r)$-system} is a set system $(X, {\cal B})$, where $|X|=n$ and $\cB \subseteq {X\choose r}$
is the set of blocks such that each $k$-subset of $X$ contains at least one of the blocks.
The {\it Tur\'an number} $T(n,k,r)$ is the minimum number of blocks in such a system.
This number is determined only for $r=2$ and some sporadic examples~(see~\cite{Sid95,Kee11} and references therein).
De Caen~\cite{deCaen.1983} proved the following general lower bound on $T(n,k,r)$.

\begin{equation}\label{bndforTurannumber}
T(n,k,r)\geq \frac{n-k+1}{n-r+1} \cdot \frac{\binom{n}{r}}{\binom{k-1}{r-1}}.
\end{equation}

The following proposition is an immediate result from the definition of Tur\'an systems.

\begin{proposition}
\label{CPCnTuransys}
A family of codesets $\{\cC_1, \cC_2,\ldots,\cC_M\}$ is an $(n,t,w)$-CPC code
if and only if the set system of each $\cC_i$ is a Tur\'an $(n,n-t,w)$-system and these $M$ set systems are pairwise disjoint.
\end{proposition}

Combining the bound in~(\ref{bndforTurannumber}) and Proposition~\ref{CPCnTuransys},
we have the following upper bound on the size of CPC codes.

\begin{theorem}\label{thm:upperbound}
If $\mC$ is an $(n,t,w)$-CPC code of size $M$, then
\begin{equation*}
M\leq \frac{n-w+1}{t+1}{n-t-1\choose w-1}.
\end{equation*}
\end{theorem}
\begin{proof}
By (\ref{bndforTurannumber}) we have that
\begin{equation}
\label{eq:Tr1}
T(n,n-t,w)\geq \frac{t+1}{n-w+1} \cdot \frac{\binom{n}{w}}{\binom{n-t-1}{w-1}}.
\end{equation}
By Proposition~\ref{CPCnTuransys} we have that
\begin{equation}
\label{eq:Tr2}
M \leq \frac{\binom{n}{w}}{T(n,n-t,w)}.
\end{equation}
Combining (\ref{eq:Tr1}) and (\ref{eq:Tr2}) yield that
$$
M\leq \frac{n-w+1}{t+1}{n-t-1\choose w-1}.
$$
\end{proof}
\begin{corollary}
\label{cor:Turan}
If $\mC$ is an $(n,t,w)$-LPC code of size $M$, then
$$
M \leq \sum_{i=0}^{w-1} {n\choose i}+\frac{n-w+1}{t+1}{n-t-1\choose w-1}~.
$$
\end{corollary}
\begin{proof}
If we consider an $(n,t,w)$-CPC code $\mC$, then to form an $(n,t,w)$-LPC code
we can add to $\mC$ at most $\sum_{i=0}^{w-1} {n\choose i}$
codesets, each one contains exactly one codeword of weight less than $w$.
\end{proof}

When $t=\Theta(n)$, we have that $(n-w+1)/(t+1)=O(1)$, and so
the upper bound for $(n,t,w)$-CPC codes is improved from $O(n^w)$ implied by Theorem~\ref{thm:upper_bounds}
to $O(n^{w-1})$ implied by Theorem~\ref{thm:upperbound}.

For an $(n,t,w)$-LPC code, we have by Corollary~\ref{cor:Turan} that
the size of such a code is at most
$$\sum_{i=0}^{w-1} {n\choose i}+\frac{n-w+1}{t+1}{n-t-1\choose w-1},$$
which is also $O(n^{w-1})$
when $t$ and $n$ are of the same order of magnitude.

\subsection{Some Known Constructions}

Chee et. al~\cite{Cheeetal.2017} provided the following construction of LPC/CPC codes.

\begin{proposition}[Decomposition of Complete Hypergraphs]
\label{prop:decomp}
If $n=(t+1)w$ then there exists an $(n,t,w)$-CPC code of size $\binom{n-1}{w-1}$.
\end{proposition}

When $w$ is fixed, we have that $t$ and $n$ are of the same order of magnitude and
the above construction attains the asymptotic upper bound $O(n^{w-1})$.
Unfortunately, usually no efficient encoding and decoding methods are known for this construction and
generally the only known encoding method involves listing all the $\binom{n-1}{w-1}$ codesets.
The exceptions are for small $n$ or when $w$ is very small, e.g. when $w=2$ or $w=3$~\cite{Bet74,DeMi02}.

Chee et. al~\cite{Cheeetal.2017} also proposed the following constructions of LPC codes
which have efficient coding schemes.

\begin{proposition}[Concatenation]\label{prop:concat}
Suppose that $q\le \sum_{i=0}^{w'}\binom{s}{i}$ and $q$ is a prime power and $t\le s$.
\begin{enumerate}[(i)]
\item If $t+1\le m/2$, then there exists an $(ms,t,mw')$-LPC code of size $q^{m-t-1}$.
\item If $t+1\le m\le q+1$, then there exists an $(ms,t,mw')$-LPC code of size $q^{m-t}$.
\end{enumerate}
\end{proposition}

\begin{proposition}[Sunflower Construction]\label{prop:sunflower}
Let $r+t\le (n+s)/2$. If a linear $[n,s,w+1]_2$ code exists and a linear $[n-t,r,w+1]_2$ code does not exist,
then there exists an $(n,t,w)$-LPC code of size $2^{n-t-r}$.
\end{proposition}

Finally, Proposition~\ref{CPCnTuransys} suggests a new method to construct $(n,t,w)$-CPC codes.
We just have to find a set with large number of pairwise disjoint Tur\'an $(n,n-t,w)$-systems.
One work in this direction was done in~\cite{EWZ95} where pairwise disjoint Tur\'an $(n,w+1,w)$-systems
were considered. Another possible construction based on Proposition~\ref{CPCnTuransys} is
to consider complements of pairwise disjoint Steiner systems. Such pairwise disjoint systems were considered
in~\cite{BSSS90,vPEt89} and for Steiner quadruple systems which will be used in the sequel in~\cite{Etzion.1991}.

\section{An Efficient Construction for Constant-Power Cooling Codes}
\label{sec:efficient}

In this section, we present a new construction of CPC codes which
has efficient encoding and decoding algorithms. Asymptotically, the
codes obtained by the construction attain the bound of Theorem~\ref{thm:upperbound}.
As was mentioned before,
there are three types of constructions for LPC codes in~\cite{Cheeetal.2017}. The first one is based
on decomposition of the complete hypergraph, the second one is a concatenation method based on $q$-ary cooling codes,
and the third one is a Sunflower Construction.
The construction in this section, is an explicit construction for CPC codes which combines the advantages
of the first two types of constructions. We first rephrase the construction based on decomposition
of the complete hypergraph in terms of set systems.
The construction is based on the following generalization of Proposition \ref{prop:decomp}.

\begin{proposition}\label{resolvable}
Let $(X,\cB)$ be a set system of order $n$, where $\cB$ is partitioned into $M$ partial parallel classes
$\cP_1, \cP_2, \ldots, \cP_M$. If $\cB \subseteq {X \choose w}$ and each $\cP_i$ has at least $t+1$ blocks,
then the codesets $\cP_1, \cP_2, \ldots, \cP_M$ form an $(n,t,w)$-CPC code.
\end{proposition}

\begin{proof}
By definition, each codeword of a codeset has weight $w$. Hence, to show that
$\cP_1, \cP_2, \ldots, \cP_M$ form an $(n,t,w)$-CPC code, we only have to prove that
given a $t$-subset $S$ of $X$ with the list of hottest wires and a codeset $\cP_i$, $1\leq i\leq M$,
there exists a block $B \in \cP_i$ such that $B \cap {S}=\varnothing$.
Since $\cP_i$ is a partial parallel class with at least $t+1$ codewords, it follows that ${S}$
intersects at most $t$ blocks of $\cP_i$.
Hence, there exists a block $B \in \cP_i$ such that $B \cap {S}=\varnothing$.
\end{proof}

The \emph{complete $k$-uniform hypergraph} $G=(V,E)$ has a vertex set $V$ with $n \geq k$
vertices, and each subset of $\binom{V}{k}$ is connected by an hyperedge.
The \emph{decomposition} of $G$ is a partition of the set of edges $E$ in $G$
into disjoint perfect matching. In other words, a partition into vertex-disjoint sets of edges, where each vertex of $V$
appears exactly once in each set of the partition.
The celebrated Baranyai's theorem~\cite[p. 536]{vLWi92} asserts that such a decomposition
always exists if $k$ divides the number of vertices in $V$.
Therefore, since a decomposition of the complete $k$-uniform hypergraph with vertex
set $X$ is a resolvable set system $(X, {X \choose k})$, if $k$ divides $|X|$, we recover Proposition~\ref{prop:decomp}.

\subsection{CPC Codes Based on Linear Codes}

Let $\cC$ be an $[N,K,D]_q$ code.
Using the codewords of $\cC$, we will show how to construct a set system with $q^{K-1}$ partial parallel classes,
each one has blocks of the same size,
and as a consequence Proposition~\ref{resolvable} yields a CPC code $\mD$.
To equip $\mD$ with efficient encoding and decoding schemes, we utilize the erasure-correcting algorithms of the linear code $\cC$.
These schemes are discussed in detail in Section~\ref{sec:encoding}.

For a set of coordinates $T$ and a vector $\bsg\in\F_q^{|T|}$,
we say that {\it $\bsg$ appears $\lambda$ times in $\cC$ at $T$} if there are $\lambda$ codewords
in~$\cC$ whose restriction on $T$ is $\bsg$. Since any two
codewords of~$\cC$ differ in at least $D$ coordinates, it follows that
they agree in at most $N-D$ positions. Hence, we have the following observations.

\begin{lemma}\label{lem:appearance}
Let $\cC$ be an $[N,K,D]_q$ code.
\begin{enumerate}[(i)]
\item For any $(N-D+1)$-subset of coordinates $T$ and any $\bsg\in \F_q^{N-D+1}$, $\bsg$ appears in at most one codeword of~$\cC$ at $T$.

\item For any $(N-D)$-subset of coordinates $T$ and any $\bt\in \F_q^{N-D}$, $\bt$ appears in at most $q$ codewords of~$\cC$ at $T$.
\end{enumerate}
\end{lemma}
\begin{proof}
$~$
\begin{enumerate}[(i)]
\item If $\bsg\in \F_q^{N-D+1}$ appears twice in codewords of~$\cC$ at an $(N-D+1)$-subset
of coordinates $T$, then the two related codewords have distance at most $D-1$, a contradiction.

\item If $\bt\in \F_q^{N-D}$ appears in $q+1$ codewords of~$\cC$ at an $(N-D)$-subset of coordinates $T$, then let
$t$ be a coordinate not in $T$. In at least two of the related codewords coordinate $t$
has the same symbol. We add this symbol to $\bt$ to obtain $\bsg\in \F_q^{N-D+1}$ which
appears in two codewords of~$\cC$ at the $(N-D+1)$-subset $T \cup \{t\}$, contradicting claim (i)
of this lemma.
\end{enumerate}
\end{proof}

For a code $\cC$ and a subset of coordinates $T$, let $\cC|_T$ denotes the set of codewords
restricted to the coordinates of $T$, i.e., the projection of $\cC$ into the set of coordinates indexed by $T$.
For a word $\vu$, let $\vu|_T$ denotes the restriction of $\vu$
to the coordinates of $T$. Finally, for a matrix $\vG$, let $\vG|_T$ denotes the submatrix of $\vG$
obtained from the columns indexed by $T$.

\begin{lemma}\label{submatrix}
Let $\cC$ be an $[N,K,D]_q$ code. If $\vG$ is a generator matrix
of $\cC$, then every $K \times (N-D)$ submatrix of $\vG$ has rank either $K$ or $K-1$.
Furthermore, there exists a $K \times (N-D)$ submatrix of $\vG$ whose rank is $K-1$.
\end{lemma}

\begin{proof}
Let $T$ be a subset of $N-D$ coordinate positions and assume the corresponding $K \times (N-D)$ submatrix~$\vG|_T$
has rank~$r$, where $r\leq K$.  Let $\phi_T$ be the linear map from $\F_q^{K}$ to $\F_q^{N-D}$
defined by $\phi_T(\vx)=\vx \vG|_T$. Clearly, the dimension of the kernel of $\phi_T$ is $K-r$.
Hence, the all-zeroes vector of length $N-D$ appears in $q^{K-r}$ codewords of~$\cC|_T$.
By Lemma~\ref{lem:appearance} the all-zeroes vector appears in at most $q$ codewords of~$\cC|_T$
at~$T$ which implies that $K-r \leq 1$ and therefore $r \geq K-1$.

Let $\vu$ be a codeword of~$\cC$ with minimum weight~$D$, $T$ be the subset of $[N]$ not in the support of $\vu$,
i.e. $T = [N] \setminus \supp(\vu)$.
Let~$\vx$ be the information vector of length~$K$ such that~$\vu = \vx \vG$.
Since $\vu$ has weight $D$, it follows that $|\supp(\vu)|=D$ and hence the size of $T$ is $N-D$.
Since $\vu$ has \emph{zeroes} in the coordinates of $T$, it follows that for the $K \times (N-D)$ submatrix $\vG|_T$
of~$\vG$ we have that $\vx \vG|_T={\bf 0}$.
Therefore, the rank of $\vG_T$ is at most $K-1$.

Thus the rank of $\vG|_T$ is at least $K-1$ by the first part of the proof and
at most $K-1$ by the second part of the proof, which implies that
the rank of $\vG|_T$ is $K-1$.
\end{proof}

We are now ready to present the new efficient construction for CPC codes.

\begin{construction}
\label{LC2CPCC}
\label{con:main}
Let $\cC$ be an $[N,K,D]_q$ code and $\vG$ be a generator matrix of $\cC$, where
the last $N-D$ columns of $\vG$ form a $K \times (N-D)$ submatrix of $\vG$ whose rank is $K-1$.

\begin{itemize}
\item Partition the codewords of $\cC$ into disjoint codesets $\cC_1,\cC_2,\ldots,\cC_M$ such that
two codewords $\vu$ and $\vv$ are in the same codeset if and only if they agree on their last $N-D$ symbols.

\item For each $i\in[M]$, truncate the codewords in $\cC_i$ to length $w$ by removing their last $N-w$ symbols.
In other words, set $\cC'_i \triangleq \{\vu|_{[w]}: \vu\in \cC_i\}$ for each $i\in[M]$.

\item For each $i \in [M]$ construct the set system
$(X,\cD_i)$, where $X=\F_q\times [w]$ and
$$
\cD_i=\{(x_j,j) ~:~ \vx=x_1x_2\cdots x_w \in \cC_i',~  j\in [w]\}.
$$
\end{itemize}
\end{construction}

\begin{theorem}
\label{thm:main_pmtr}
If $N-D+1\leq w\leq D$, then the collection of codesets $\mD=\{\cD_1,\cD_2,\ldots, \cD_M\}$
is an $(n,t,w)$-CPC code of size $M=q^{K-1}$, where $n=qw$ and $t = q-1$.
\end{theorem}
\begin{proof}
Clearly, by the definition of the construction we have that $n=qw$ and each codeword
has weight $w$. Hence, to complete the proof it is sufficient to show that $M=q^{K-1}$, the $M$ codesets are pairwise disjoint,
and for any $t$-subset $S$ of coordinates from $X$ and each $i \in [M]$, there exists a codeword $\vu$
in the codeset $\cD_i$ such that $\supp(\vu)\cap S = \varnothing$.

\begin{enumerate}
\item Let $\vG'$ be the $K \times (N-D)$ submatrix of $\vG$ formed from the last $N-D$ columns of $\vG$.
Consider the linear map $\phi$ from $\F_q^K$ to $\F_q^{N-D}$ defined by $\phi(\vx)=\vx \vG'$.
Since the rank of $\vG$ is $K-1$, it follows that the image of~$\phi$ has dimension $K-1$ and
the kernel of $\phi$ has dimension one. It follows that $M=q^{K-1}$ and $|\cC_i|=q$ for each $i \in [M]$.

\item The minimum distance of $\cC$ is $D$ and hence each two codewords of $\cC$
can agree in at most $N-D$ coordinates, i.e. they differ in any subset of $N-D+1$ coordinates.
Since $w \geq N-D+1$ and the codewords of $\cC$ were shortened in their last $N-w$ coordinates,
it follows that all the shortened codewords of~$\cC$ are distinct.
Thus $\cC_1',\cC_2', \ldots, \cC_M'$ are pairwise disjoint and $|\cC_i'|=|\cC_i|=q$.
Now, it can be easily verified by the definition of~$\cD_i$ that
$\cD_1,\cD_2, \ldots, \cD_M$ are pairwise disjoint and $|\cD_i|=q$ for each $i \in [M]$.

\item Each two codewords of $\cC_i$ agree on their last $N-D$ coordinates and since their distance
is at least $D$, it follows that they differ in the first $D$ coordinates. Since $w \leq D$, this implies that
any two codewords of~$\cC_i'$ differ in all their $w$ coordinates. Hence, by the definition
of~$\cD_i$ it implies that each two codewords of~$\cD_i$ differ in their nonzero coordinates.
Therefore, the codewords in~$\cD_i$ are pairwise disjoint, i.e.,
$\cD_i$ is a partial parallel class. We also have that $|\cD_i|=q$ for each $i \in [M]$.
Hence, $\cD_i$ is a parallel class and as in the proof of Proposition~\ref{resolvable} we have that for any $t$-subset $S$,
$\cD_i$ has a codeword $\vu$ such that $\supp(\vu)\cap S = \varnothing$.
\end{enumerate}
Thus, the the required claims were proved and hence the collection of codesets
$\mD=\{\cD_1,\cD_2,\ldots, \cD_M\}$ is a $(qw,q-1,w)$-CPC code of size $M=q^{K-1}$.
\end{proof}

For a given $[N,K,D]_q$ code $\cC$ and its generator matrix $\vG$ in Construction~\ref{con:main},
we need to find a minimum weight codeword in $\cC$ in order to determine a $K \times (N-D)$-submatrix of $\vG$
with rank $K-1$, i.e., to find a permutation of the columns of $\vG$ such that the last $N-D$ coordinates of $\vG$
will have rank $K-1$.
Finding the minimum distance of a code is an NP-hard problem and the decision
problem is NP-complete~\cite{Var97}. Therefore,
we focus on certain families of codes where it is computationally easy
to find minimum weight codewords. One such family is the {\em maximum distance separable (MDS)} codes.
Recall that a linear $[N,K,D]_q$ code is an MDS code if $D=N-K+1$~\cite[Ch.11]{MacWilliams.1977}.
If the code $\cC$ in Construction~\ref{con:main} is an MDS code, then every $K$ columns of $\vG$ are
linearly independent and hence each $K\times (N-D)$ submatrix of $\vG$ has rank $K-1$ since $N-D=K-1$.
Therefore, we may use any $N-D$ coordinate as the last $N-D$ coordinates of $\cC$.
It is well known that MDS codes exist for the following parameters.

\begin{theorem}[{see~\cite[Ch.11]{MacWilliams.1977}}]
\label{thm:MDS}
Let $q$ be a prime power. If $D\geq 3$, then there exists an $[N,K,D]_q$ MDS code
if $N\leq q+1$ for all $q$ and $2\leq K\leq q-1$, except when $q$ is even and $K\in\{3,q-1\}$, in which case $N\leq q+2$.
\end{theorem}

Setting $N=q+1$, $K=w$, $D=q-w+2$, and using an $[N,K,D]_q$ MDS code as the code $\cC$ in Construction~\ref{con:main},
we have that $w \leq D =q-w+2$, i.e., $q \geq 2w-2$.
Hence, Theorem~\ref{thm:main_pmtr} yields the following corollary.

\begin{corollary}
\label{corodirect}
Let $n,t$ and $w$ be positive integers.
If $q=n/w$ is a prime power and $q \geq 2w-2$, then there exists an $(n,q-1,w)$-CPC code of size
$({n}/{w})^{w-1}$.
\end{corollary}


In Corollary~\ref{corodirect}, when $w$ is fixed, $t=q-1$ has the same order of magnitude as $n$.
Hence, the codes constructed in this case asymptotically attain the upper bound $O(n^{w-1})$.
We also note that for some parameters, these CPC codes are much larger than the LPC codes provided by
Propositions \ref{prop:concat} and \ref{prop:sunflower}.

\begin{example} By choosing $n=96$, $w=6$ and $t=15$, Corollary~\ref{corodirect} yields a $(96,15,6)$-CPC code of size $16^5=2^{20}$.

In contrast, suppose we use Proposition \ref{prop:concat} to construct a $(96,t,6)$-LPC code with $t\le 15$.
The largest size $16^5=2^{20}$ is obtained by choosing $m = 6$, $t = 1$, $s = 16$, $w'=1$, and $q=16$.
The resulting $(96,1,6)$-LPC code has the same size as the CPC obtained by Corollary~\ref{corodirect},
but the cooling capability of the former is clearly much weaker.
Proposition \ref{prop:sunflower}, on the other hand, yields a $(96,15,6)$-LPC code of size $2^{16}$ by choosing $s=81$ and $r=65$.
This code has similar parameters, but its size is much smaller.
\end{example}


\begin{example}
If we choose a $[17,8,9]_9$ code (see \cite{Grassl}).
If $w=9$ and $t=8$ in Construction~\ref{con:main}, then we obtain an $(81,8,9)$-CPC code of size $9^7\approx 2^{22.189}$.

In contrast, the largest $(81,8,9)$-LPC code obtained from Proposition \ref{prop:concat} has size $9\approx 2^{3.17}$
by choosing ${m = s = q= 9}$, $w'=1$.
Proposition \ref{prop:sunflower}, on the other hand, yields an $(81,8,9)$-LPC of size $2^{21}$ by choosing $s=54$ and $r=52$.
\end{example}

\subsection{Encoding and Decoding Schemes}
\label{sec:encoding}
{
We continue in this subsection and discuss the encoding and decoding schemes for the code $\mD$
obtained in Construction~\ref{LC2CPCC}.
Let $\vG$ be a generator matrix of the $[N,K,D]_q$ code $\cC$,
where the last $N-D$ columns
of $\vG$ form a $K \times (N-D)$ submatrix $\vG'$ whose rank is $K-1$. Furthermore, w.l.o.g. we assume that $\vG$ has the form
\begin{equation*}
\vG=
\left(
  \begin{array}{cc}
    \vA & \vI_{K-1}\\
   \bta_K & 0 \cdots 0\\
  \end{array}
\right),
\end{equation*}
where $\vI_{K-1}$ is the identity matrix of order $K-1$.

Each codeset in $\mD$ will be identified by the unique vector from $\F_q^{K-1}$.
This is possible since the number of codesets is $q^{K-1}$.
For $\bsg\in\F_q^{K-1}$, let $\cC_{\bsg}$ be the set of $q$ codewords from $\cC$ whose suffix of length $K-1$ is $\bsg$.
Furthermore, let $\cC_{\bsg}'$ and $\cD_{\bsg}$ be the derived codesets as defined in Construction~\ref{LC2CPCC}.

Given a $t$-subset $S$ of $\F_q\times[w]$ and a word $\bsg=(\sigma_1,\sigma_2,\ldots,\sigma_{K-1})\in \F_q^{K-1}$,
our objective for encoding of Construction~\ref{LC2CPCC} is to find a codeword $\vu \in \cD_{\bsg}$ such that $\supp (\vu) \cap S=\varnothing$.
Let $\bta_i$ be the $i$-th row of $\vG$. Let
$$\vr= \bsg \vA|_{[w]} = \sum_{i=1}^{K-1} \sigma_i \bta_i|_{[w]} ,$$
and hence the codeset $\cC'_{\bsg}$ is
$$\cC'_{\bsg}=\{\vr+\lambda \bta_K|_{[w]}: \lambda\in \F_q\}.$$
The codeset $\cD_{\bsg}$ is derived from $\cC'_{\bsg}$ as indicated in Construction~\ref{LC2CPCC}, and hence
we can consider the intersection of each one of the $q$ blocks in $\cD_{\bsg}$ with $S$
to find the block $B$ such that $B \cap S=\varnothing$.

Hence, for the encoding, $O(n)$ multiplications over $\F_q$ are required to find $\cD_{\bsg}$.
During this computation we can also check whether each codeword
of $\cD_{\bsg}$ has nontrivial intersection with $B$ or not. Therefore,
there is no need for further computations to find $B$.

For the decoding, suppose that we have a codeword $\{(x_1,1),(x_2,2),\ldots,(x_{w},w)\}$.
By our choice we have that $w \geq N-D+1$ which implies that $D-1 \geq N-w$ and hence we can correct any $N-w$
erasures in any codeword of~$\cC$. Hence, the $N-w$ erasures in $(x_1,x_2,\ldots,x_{w}, ?, ?,\ldots,?)$ can be recovered
and the last $K-1$ symbols, $x_{N-K+2}, x_{N-K+3},\ldots,x_N$ are the information symbols.
In particular, if the code $\cC$ is a Reed-Solomon code, then by using Lagrange interpolation,
$O(w^3)$ multiplications are enough to perform the decoding, e.g.~\cite{RoRu00}.

\section{Error-Correcting CPC Codes}
\label{sec:ECC}
In this section we consider CPC codes that can correct transmission errors (`0' received as `1', or `1' received as `0').
An $(n,w,t)$-CPC which can correct up to $e$ errors will be called an $(n,t,w,e)$-CPECC (constant weight power error-correcting cooling) code.
First, Construction~\ref{LC2CPCC} will be used to produce CPECC codes by examining the minimum distance
of the constructed codes.

\begin{theorem}
\label{MDS2CPECC}
{
If the code $\cC$ used for Construction~\ref{LC2CPCC} is an $[N,K,D]_q$ code,
then the code $\mD$ obtained by Construction~\ref{LC2CPCC} is an $(n,t,w,e)$-CPECC code of size $M=q^{K-1}$,
where $n=qw$, $t= q-1$, and $e \geq w+D-N-1$.}
\end{theorem}
\begin{proof}
All the parameters of the code except for $e=w+D-N-1$ were proved in Theorem~\ref{thm:main_pmtr}.
Since the minimum distance of $\cC$ is $D$ and the code $\cC$ was punctured in the last
$N-w$ coordinates to obtain the code $\cC'$ (the union of the codesets $\cC_i'$, $1 \leq i \leq M$),
it follows that the minimum distance of $\cC'$ is at least $D-(N-w)$. By the definition of $\cD'$ (the union
of the codesets $\cD_i'$, $1 \leq i \leq M$) we have that if $\vu, \vu' \in \cC'$ differ in $\ell$
coordinates, then the related codewords in $\cD$ differ in $2 \ell$ positions. Hence, the minimum distance of $\cD$
is at least $2(D+w-N)$ and thus the number of errors that it can correct is $e \geq w+D-N-1$.
\end{proof}

Next, an algorithm which demonstrates the error-correction for an $(n,t,w,e)$-CPECC code
will be given. For simplicity, we will focus on a special example,
where our starting point is a Reed-Solomon code $\cC$ (which is of course
an MDS code), where $K=N-D+1=w-e$.

\begin{construction}
\label{RSECCcon}
Let $w$ and $e$ be  positive integers and $q$ be a prime power such that $q\geq 2w-e-1$.
Let $a_1,a_2, \ldots, a_{w},b_1,b_2,\ldots, b_{w-e-1}$ be $2w-e-1$ distinct elements of $\F_q$.
\begin{itemize}
\item For each polynomial $f(X)\in\F_q[X]$, define the following block on the point set $\F_q \times [w]$,
\[C_f=\{(f(a_j),j): j\in [w],~\deg(f)\leq w-e-1\}.\]
\item For each $\bsg=(\sigma_1,\sigma_2,\ldots, \sigma_{w-e-1}) \in \F_q^{w-e-1}$, let
\begin{align*}\
\cE_{\bsg}=\{&C_f:  f \in \F_q[X],~ \deg(f)\leq w-e-1,~ f(b_i)=\sigma_i \mbox{ for each }i \in[w-e-1]\}.
\end{align*}
\end{itemize}

\end{construction}

\begin{theorem}
\label{thm:ecc}
The code $\mE= \{ \cE_{\bsg} ~:~ \bsg \in \F_q^{w-e-1} \}$ is an $(n,t,w,e)$-CPECC code
of size $q^{w-e-1}$, where $n=qw$ and $t=q-1$.
\end{theorem}
\begin{proof}
It is an immediate observation from the definition of the point set $\F_q \times [w]$ and
the codeword $C_f$ that each codeword has length $qw$ and weight $w$.
The rest of the proof has four steps. In the first one we will prove that
for each $\bsg, \bsg' \in F_q^{w-e-1}$,
$\cE_{\bsg}$ and $\cE_{\bsg'}$ are disjoint whenever $\bsg\not=\bsg'$.
In the second step we will prove that for each $\bsg \in \F_q^{w-e-1}$ the blocks in $\cE_{\bsg}$ are pairwise
disjoint. As a result, by a simple counting argument in the third step it will be proved
that $\mE$ has $q^{w-e-1}$ codesets, each one has parallel class of size $q$, and
as a consequence $\mE$ is a $(qw,q-1,w)$-CPC code.
In the last step we will find the minimum Hamming distance of $\mE$ and as a result the number
of errors $e$ that it can correct.

\begin{enumerate}
\item Assume that there exist two codewords $C_f \in \cE_{\bsg}$ and $ C_g\in \cE_{\bsg'}$ such that
$\bsg \neq \bsg'$ and $C_f=C_g$.
Then $f$ and~$g$ agree on at least $w$ points and since the degrees of the polynomials are less than $w$,
it follows that $f=g$. It implies that
$\sigma_i = f(b_i) =g(b_i)=\sigma'_i$ for all $i \in [w-e-1]$ and hence $\bsg=\bsg'$, a contradiction.
Thus, $\cE_{\bsg}$ and~$\cE_{\bsg'}$ are disjoint whenever $\bsg\not=\bsg'$.

\item Assume that the blocks $C_f$ and $C_g$ in $\cE_{\bsg}$, where $f \neq g$, intersect at
the point $(x,i_0)$ for some $x\in \F_q$ and $i_0\in[w]$. It implies that $f(a_{i_0})=g(a_{i_0})$ and since
$C_f, C_g \in \cE_{\bsg}$, it follows that $f(b_i)=g(b_i)$ for each $i \in [w-e-1]$.
Therefore, $f$ and $g$ agree on at least $w-e$ points. Since the degrees of $f$ and $g$  are at most $w-e-1$,
it follows that $f=g$, a contradiction.  Therefore, the blocks in $\cE_{\bsg}$ are pairwise disjoint.
Recall that each block has size $w$ and the size of the point set of these blocks $\F_q \times [w]$ is $qw$.
Hence, each codeset $\cE_{\bsg}$ contains at most $q$ blocks.

\item The number of distinct polynomials in $\F_q[X]$ whose degrees are at most $w-e-1$ is $q^{w-e}$.
Each polynomial induces exactly one codeword in $\mE$.
Hence, $\mE$ contains exactly $q^{w-e}$ distinct codewords.
Since there are $q^{w-e-1}$ codesets and each one contains at most $q$ codewords,
it follows that each one contains exactly $q$ codewords. The length of a codeword is $qw$ and
the weight of a codeword is $w$ which implies that each codeset is a parallel class.
Thus, by Proposition~\ref{resolvable}, $\mE$ is a $(qw,q-1,w)$-CPC code.

\item Finally, for any two distinct codewords $C_f$ and $C_g$, where $f$ and $g$ have degree at most $w-e-1$,
we have that $|C_f\cap C_g| \leq w-e-1$ since larger intersection implies that $f = g$.
Therefore, the Hamming distance between $C_f$ and $C_g$ is at least $2e+2$.
Thus, the code $\mE$ has minimum Hamming distance at least $2e+2$ and it can correct $e$ errors.
\end{enumerate}

Thus, $\mE$ is an $(n,t,w,e)$-CPECC code
of size $q^{w-e-1}$, where $n=qw$ and $t=q-1$.
\end{proof}

The encoding scheme in Section~\ref{sec:encoding} can be easily adapted for the encoding of the CPECC code $\mE$. Algorithm~\ref{algo3}
illustrates the decoding scheme for the $(n,t,w,e)$-CPECC code $\mE$ obtained in Construction~\ref{RSECCcon}.

\begin{algorithm}
\caption{Error-Correction for the CPECC codes in Construction~\ref{RSECCcon}}
\label{algo3}
\begin{algorithmic}[1]
\REQUIRE a binary word $\vu \subset \F_q \times [w]$ ~~~~~~~~ \{the word received after the transmission of a codeword\}
\ENSURE a message $\bsg \in \F_q^{w-e-1}$  ~~~~~~~~~~~~~~~~~~\{the information word that was sent\}
\FOR{each $i\in[w]$}
\IF{ $|Y_i \triangleq \{(y,i): (y,i)\in \vu \}|=1$}
\STATE $y_i\gets y$, where $(y,i)$ is the unique pair in $Y_i$;
\ELSE
\STATE $y_i\gets$ `?';
\ENDIF
\ENDFOR
\STATE $\hat{\vy} \gets (y_1,y_2,\ldots, y_{w})$;
\STATE apply the decoding algorithm for Reed-Solomon codes on $\hat{\vy}$
\STATE The output of the algorithm is a polynomial $L(x)$ of degree $w-e-1$;
\STATE  $\bsg \gets (L(b_1), L(b_2), \ldots, L(b_{w-e-1}))$;
\RETURN $\bsg$;
\end{algorithmic}
\end{algorithm}

\begin{theorem}
Suppose that the codeword $\vc \in \mE$ obtained in Construction~\ref{RSECCcon} was submitted and
the word $\vu$ was received from $\vc$ with at most $e$ errors.
Then, Algorithm~\ref{algo3} returns the word $\bsg \in \F_q^{w-e-1}$ such that $\vc \in \cE_{\bsg}$.
\end{theorem}
\begin{proof}
Using the notation of Algorithm~\ref{algo3}, let $i\in [w]$, $Y_i \triangleq \{(y,i): (y,i)\in \vu\}$,
and $e' =| \{ i ~:~ |Y_i|\neq 1 \}|$. If $|Y_i|=0$ then an erasure occurred and this is reflected
as an erasure in $y_i$. If $|Y_i| > 1$ then
we also know that an error has occurred for at least one coordinate $(y,i)$. This will be also reflected
as an erasure in $y_i$. Hence, at least $e'$ erasure errors are reflected in $\hat{\vy}$ as a result
of at least $e'$ errors in these $Y_i$'s.
For the remaining $w-e'$ $Y_i$'s, while there may be errors, we know that each of these $Y_i$'s contains either
no errors or two errors.
Thus, the number of other erroneous $Y_i$'s is at most $\lfloor{(e-e')}/{2}\rfloor$.

The vector $\hat{\vy}$ is obtained by mapping the subsets $Y_1,Y_2,\ldots, Y_w$ to the elements of $\F_q\cup \{ ? \}$.
The word $\hat{\vy}$ was obtained from a codeword $\vx_f$ of a Reed-Solomon code
of length $N=w$, dimension $K=w-e$, and minimum Hamming distance $D=N-K+1=e+1$.
An error-correction algorithm for such a code is capable of correcting $e'$ erasures and
at most $\lfloor{(e-e')}/{2}\rfloor$ errors as required by Algorithm~\ref{algo3}.
\end{proof}

Using the Berlekamp-Welch algorithm~\cite{WB.1986} we can correct the errors with $O(q^3)$
operations~\cite{WB.1986}, and hence, Algorithm~\ref{algo3} has complexity $O(n^3)$.

\section{Recursive Construction}
\label{sec:recursion}

All the $(n,t,w)$-CPC codes obtained from  Proposition~\ref{prop:decomp} and Construction~\ref{con:main}
have $t= n/w -1$. In this section, we present a recursive construction that yields $(n,t,w)$-CPC codes
which will be designed especially for larger values of $t$.

For this purpose, recall the conditions of Proposition~\ref{resolvable}.
Let $(X,\cB)$ be a set system with a point set of size~$n$, where $\cB \subseteq {X \choose w}$ and
$\cB$ can be partitioned into $M$ partial parallel classes $\cP_1, \cP_2, \ldots, \cP_M$ to
form a code $\mE$ with $M$ codesets.
Suppose further that each partial parallel class $\cP_i$ has exactly $q$ blocks.
Let $S$ be a $t$-subset of $X$ and $\cP_i$ be a given partial parallel class.
If $t \geq q$, it might not be possible to choose a block/codeword in $\cP_i$ which avoids~$S$.
However, by the pigeonhole principle, we can find such a block/codeword which intersects~$S$ in at most $\floor{{t}/{q}}$ elements.
Given a $(w,\floor{t/q},w')$-LPC code $\mC$ it is possible to substitute it instead of each block/codeword
of $\cB$ and break up each codeword into codewords of weight at most $w'$.
This will enable to find a block/codeword of weight $w'$ which avoids $S$.
The following construction is based on this idea, where the code $\mE$ is constructed
similarly to the code in Construction~\ref{RSECCcon}.

\begin{construction}
\label{recursivecon}
Let $q \geq n+w-1$ be a prime power and let
$a_1,a_2, \ldots, a_{n}$, $b_1,b_2,\ldots, b_{w-1}$ be $n+w-1$ distinct elements of $\F_q$.
\begin{itemize}
\item Consider the point set $\F_q\times [n]$ and let
$$\cB=\{C_f \triangleq \{(f(a_j),j): j\in [n]\}: f \in \F_q[x],~ \deg(f)\leq w-1\}.$$
Note that the size of each block $C_f$ is $n$.

\item For each $\bsg=(\sigma_1,\sigma_2,\ldots, \sigma_{w-1}) \in \F_q^{w-1}$, let
\begin{align*}\
\cE_{\bsg}=\{&C_f:  f \in \F_q[X],~ \deg(f)\leq w-1,~ f(b_i)=\sigma_i \mbox{ for each }i \in[w-1]\}.
\end{align*}
Similarly to the proof of Theorem~\ref{thm:ecc} one can show that
$\cB$ is partitioned by $\cE_{\bsg}$, $\bsg \in \F_q^{w-1}$ into $q^{w-1}$ parallel classes,
each one of size $q$.
Label the parallel classes and their blocks by $\cP_i =\{B_{ij}: j\in [q]\}$ for $i\in[q^{w-1}]$.

\item Let $\mD$  be an $(n,t,w)$-CPC code of size $m$, where $t \geq n/w$.

\item Each block $B_{ij}$ is replaced by the codewords of each codeset of $\mD$ by using any bijection between
the set of points of $B_{ij}$ and the point set of $\mD$.
Therefore, each codeword in $\mD$ corresponds to a $w$-subset of $B_{ij}$
and from each block $B_{ij}$ we construct codewords for $m$ new codesets. These sets of codewords from the $m$ codesets
will be denoted by $\cE_{ij\ell}$ for each $\ell\in [m]$.

\item For $(i,\ell)\in [q^{w-1}]\times[m]$, the codeset $\cE_{i\ell}$
is defined by $\cE_{i\ell} \triangleq \bigcup_{j=1}^q \cE_{ij\ell}$.
\end{itemize}
\end{construction}

Along the same lines of the proof in Theorem~\ref{thm:ecc} one can prove that
\begin{theorem}
The code $\{ \cP_i ~:~ 1 \leq i \leq q^{w-1} \}$ is an $(nq,q-1,n)$-CPC code.
\end{theorem}

\begin{theorem}
The code $\mE=\{\cE_{i\ell}: i \in [q^{w-1}],~ \ell\in [m]\}$ is an $(nq,tq,w)$-CPC code of size $mq^{w-1}$.
\end{theorem}

\begin{proof}
The size of $\mE$, the length of its codewords and their weight follow immediately from the definition of the codewords in $\mE$.

Given a $(tq)$-subset $S \subset \F_q\times [n]$ and a codeset $\cE_{i\ell}$, $(i,\ell)\in [q^{w-l}]\times[m]$,
we should find a codeword $\vu \in \cE_{i\ell}$ such that $\supp( \vu ) \cap S = \varnothing$.
Since $\cE_{i \ell}$ was constructed from the $q$ blocks of $\cP_i$ in which the codewords
of the $\ell$-th codeset of $\mD$ were substituted, we have to find first
a block $B_{ij} \in \cP_i$ which contains a subset $S'$ of $S$ whose size is at most $t$. Such a block exists
since the number of blocks in $\cP_i$ is $q$ and $S$ has size $tq$.
Since $\cE_{ij\ell}$ is a codeset in an $(n,t,w)$-CPC code, we can find a block $\vu$ in $\cE_{ij\ell}$
which avoids $S'$. As a consequence  $\supp( \vu ) \cap S = \varnothing$ as required.

To complete the proof we have to show that all the codesets of $\mE$ are pairwise disjoint, i.e.
$\cE_{i\ell}$ and $\cE_{i'\ell'}$ are disjoint whenever $(i,\ell)\ne(i',\ell')$. To this end,
it suffices to show  $\cE_{ij\ell}$ and $\cE_{i'j'\ell'}$ are disjoint for any $j,j' \in [q]$. If $(i,j)\not=(i',j')$, it can
be verified that $|B_{ij} \cap B_{i'j'}|\leq w-1$ since intersection of size $w$ will imply that the
related polynomials are equal. Hence, since each $\cE_{ij\ell}$ is a collection of $w$-subsets
of $B_{ij}$, we have that $\cE_{ij\ell}$ and $\cE_{i'j'\ell'}$ are disjoint.
If $(i,j)=(i',j')$ then $\cE_{ij\ell}$ and $\cE_{ij\ell'}$ are from the same $(n,t,w)$-CPC code and therefore they are disjoint.
\end{proof}

Construction~\ref{recursivecon} can be applied also on $(n,t,w)$-LPC code (instead of $(n,t,w)$-CPC code).
The only condition is that there is no codeset in which there are codewords of different weight.
Also, when there are codewords of weight $w' < w$ in the codeset, the whole construction should work with $w'$
instead of $w$, e.g. the degree of the polynomial must be at most $w'-1$.

\begin{corollary}
\label{cororec}
Let $q$ be a prime power. If  $t+w\leq n$ and  $q \geq n+w-1$, then
\begin{enumerate}[(i)]
\item there exists an $(nq,tq,w)$-CPC code of size $q^{w-1}$;
\item there exists an $(nq,tq,w)$-LPC code of size $\sum_{i=0}^{w-1} q^i$.
\end{enumerate}
\end{corollary}
\begin{proof}
\begin{enumerate}[(i)]
\item the first claim follows from the fact that we can use an $(n,t,w)$-CPC code with exactly one codeset
which contains all the $w$-subsets of the related $n$-set.

\item the second claim follows from the fact we can apply Construction~\ref{recursivecon} and claim (i) on any $w' \leq w$
and obtain disjoint codes that can be combined together.
\end{enumerate}
\end{proof}

\begin{example}
We compare certain CPC codes obtained from Construction~\ref{recursivecon} and
Corollary ~\ref{cororec} with the LPC codes obtained from Proposition~\ref{prop:sunflower}.
\begin{enumerate}[(i)]
\item Consider the set of five disjoint $3$-$(10,4,1)$ designs constructed by Etzion and Hartman~\cite{Etzion.1991}.
By taking the complements of the blocks we obtain a $(10,3,6)$-CPC code of size five.
Applying Construction~\ref{recursivecon} with $q=16$, we obtain a $(160,48,6)$-CPC code of size $5\cdot 16^5\approx 2^{22.322}$.

In contrast, Proposition \ref{prop:sunflower} yields a $(160,48,6)$-LPC code of size $2^{17}$ by setting $s=137$ and $r=95$.

\item Setting $n=9$, $t=2$, $w=7$, and $q=16$ in Corollary~\ref{cororec}  yields a $(144,32,7)$-LPC code of size ${\sum_{i=0}^6 16^i\approx 2^{24.093}}$.

In contrast, Proposition \ref{prop:sunflower} yields a $(144,32,7)$-LPC code of size $2^{18}$ by setting $s=121$ and $r=94$.
\end{enumerate}
\end{example}

{
In the regime where $w$ is fixed and $t$ has order of magnitude as $n$, we show that the
codes obtained in this section are asymptotically larger than those obtained from Proposition \ref{prop:sunflower}.
The CPC codes obtained from Construction~\ref{recursivecon} and Corollary~\ref{cororec} attain
the asymptotic upper bound $O((nq)^{w-1})$ when $w$ is fixed. In contrast,
if we apply Proposition~\ref{prop:sunflower} with $s= nq-\lceil\log_2 (\sum_{i=0}^{w-1} { nq-1 \choose i }\rceil  )$
(the Gilbert-Varshamov lower bound) and  $r= nq-tq-\lfloor \log_2 (\sum_{i=0}^{\frac{w}{2}} { nq -tq\choose i })\rfloor$
(the Hamming upper bound), we obtain an $(nq,tq,w)$-LPC code of smaller size
$O((nq)^{{w}/{2}})$, or $o((nq)^{w-1})$.
}

\section{LPC Codes from Cooling Codes}
\label{sec:map}

In this section we use a novel method to transform cooling codes
into low-power cooling codes, while preserving the efficiency of the
cooling codes. The construction is based on an injective mapping called domination mapping
which was defined as follows in~\cite{CEKV18}

The \emph{Hamming ball of radius $w$ in $\{0,1\}^n$} is the set $\cB(n,w)$
of all words of weight at most $w$. Explicitly,
\smash{$\cB(n,w) \triangleq \bigl\{ \yyy \in \{0,1\}^n \,:\, \wt(\yyy) \le w \bigr\}$}.
Given $m \le n$,~we are interested in injective mappings $\varphi$
from $\{0,1\}^m$ into $\cB(n,w)$ that establish a certain domination
relationship between positions in $\xxx \,{\in}\, \{0,1\}^m$ and positions
in its image $\yyy = \varphi(\xxx)$.
Specifically, one should be able to ``switch off'' every position
$j \,{\in}\, [n]$ in $\yyy$ (that is, ensure that $y_j = 0$) by switching
off a corresponding position $i \,{\in}\, [m]$ in $\xxx$ (that is, setting
$x_i = 0$). More precisely, let $G = \bigl([m] \cup [n], E\bigr)$ be a
bipartite graph with $m$ left vertices and $n$~right vertices.
If $G$ has no isolated right vertices, we refer to $G$ as
a \emph{domination graph}.

\begin{definition}
\label{domination-def}
Given an injective map $\varphi: \{0,1\}^m \to \cB(n,w)$ and
a graph $G = \bigl([m] \cup [n], E\bigr)$, we say that~$\varphi$
is a \emph{$G$-domination mapping}, or \emph{$G$-dominating} in brief, if
\begin{align*}
\forall\, (x_1,x_2,\ldots,x_m) \in \{0,1\}^m,
\hspace{0.90ex}
\forall\, (i,j) \in E :
 \hspace{36.00ex}\\
\text{if\/ $\varphi(x_1,x_2,\ldots,x_m) = (y_1,y_2,\ldots,y_n)$ and $x_i = 0$,\,
then $y_j = 0$}
\end{align*}
We say that $\varphi$ is an \emph{$(m,n,w)$-domination mapping} if
there exists a domination
graph $G = \bigl([m] \cup [n], E\bigr)$, such that $\varphi$ is {$G$-dominating}.
\end{definition}

Properties of domination mappings, bounds on their parameters, constructions, and existence theorems
were given in~\cite{CEKV18}. For our purpose we need some results from~\cite{CEKV18} and some which
will be developed in the sequel. The first one taken from~\cite{CEKV18} restricts the structure
of the domination graph.

\begin{lemma}
\label{lem:sub_G}
The domination graph $G = \bigl([m] \cup [n], E\bigr)$ of an $(m,n,w)$-domination mapping
has  a subgraph with no isolated vertices and the degrees of the right vertices is exactly one.
\end{lemma}

In view of Lemma~\ref{lem:sub_G} we will assume in the sequel that our domination graphs
have no isolated vertices and all the right vertices have degree exactly one.
We will define the \emph{neighbourhood} of a vertex $\vv$ is $G$ as the set of vertices adjacent to $\vv$
and denote it by $N(\vv)$. The following lemma is an immediate consequence of these observations
and definition.

\begin{lemma}
\label{lem:neighbour}
If $U \subset [n]$ is a set of right vertices of $G$ then $N(U) \triangleq \{ N(\vu) ~:~ \vu \in U \}$ is
a set of vertices in $[m]$ and $|N(U)| \leq |U|$.
\end{lemma}

Next, the obvious connection between domination mappings, cooling codes, and
low-power cooling codes is given in the following theorem.

\begin{theorem}
If there exists an $(m,t)$-cooling code $\mC =\{\cC_1,\cC_2,\ldots,\cC_M\}$ and an $(m,n,w)$-domination mapping
$\varphi$ then the code $\mC' = \{ \cC_1',\cC_2',\ldots,\cC_M' \}$, where
$$\cC_i' \triangleq \{ \varphi (\xxx) ~:~ \xxx \in \cC_i \},~\text{for each}~1 \leq i \leq M,$$
in an $(n,t,w)$-LPC code.
\end{theorem}
\begin{proof}
The length $n$ and the weight which is smaller from or equal to $w$ for the codewords of $\mC$ are immediate
consequences from the definition of the $(m,n,w)$-domination mapping. Now, suppose we are given a $t$-subset $S' \subset [n]$ and a
codeset $\cC_i'$ for some $1 \leq i \leq M$. To complete the proof we have to show
that there exists a codeword $\vu' \in \cC_i'$ such that $\supp(\vu') \cap S'=\varnothing$.
The $t$-subset $S'$ can be viewed as a set of right vertices in the domination graph $G = \bigl([m] \cup [n], E\bigr)$.
By Lemma~\ref{lem:neighbour}, for the set of neighbours of $S' \subset [n]$, $S \triangleq N(S') \subset [m]$, we have that $|S| \leq |S'|$
and hence $|S| \leq t$. Since $\mC$ is
an $(m,t)$-cooling code, it follows that there exists a codeword $\vu$ in $\cC_i$
such that $\supp(\vu) \cap S=\varnothing$ which implies by the domination
property that $\supp(\varphi(\vu)) \cap S'=\varnothing$.
\end{proof}

A product construction for domination mappings was presented in~\cite{CEKV18}.

Let
$\varphi_1\!: \{0,1\}^{m_1} \to \cB(n_1,w_1)$
and
$\varphi_2\!: \{0,1\}^{m_2} \to \cB(n_2,w_2)$
be arbitrary domination mappings. Then their
\emph{product}
$\varphi = \varphi_1 \times \varphi_2$\, is a mapping from
$\{0,1\}^{m_1+m_2}$ into $\cB(n_1\kern-1pt+n_2,w_1\kern-1pt+w_2)$ defined as
follows:
$$
\varphi(\xxx_1,\xxx_2)
\ = \
\bigl(\varphi_1(\xxx_1),\varphi_2(\xxx_2)\bigr)
$$
where $\xxx_1\,{\in}\,\{0,1\}^{m_1}$, $\xxx_2\,{\in}\,\{0,1\}^{m_2}$,
and $(\cdot,\cdot)$ stands for string concatenation. 
That is, in order to find the~image of a word $\xxx \in \{0,1\}^{m_1+m_2}$
under $\varphi$, we first parse $\xxx$ as $(\xxx_1,\xxx_2)$, then
apply $\varphi_1$ and $\varphi_2$ to the two parts. 
\begin{theorem}
\label{thm:product}
If $\varphi_1$ is an $(m_1,n_1,w_1)$-domination mapping
and $\varphi_2$ is an $(m_2,n_2,w_2)$-domination mapping,
then their product $\varphi = \varphi_1 \times \varphi_2$\kern1pt\ is
an $(m_1\kern-1pt+ m_2, n_1\kern-1pt+n_2,w_1\kern-1pt+w_2)$-domination mapping.
\vspace{-0.54ex}
\end{theorem}

The idea in Theorem~\ref{thm:product} can be generalized as follows to a large number of domination mappings.
\begin{theorem}
\label{thm:gen_product}
Let $\varphi_i$ be an $(m_i,n_i,w_i)$-domination mapping for each $1 \leq i \leq \ell$,
and let $(\xxx_1,\xxx_2,\ldots,\xxx_\ell)$ be a binary word, where the length of $\xxx_i$ is~$m_i$,
for each $1 \leq i \leq \ell$.
The mapping $\varphi$, defined by
$$
\varphi (\xxx_1,\xxx_2,\ldots,\xxx_\ell) = (\varphi_1 (\xxx_1),\varphi_2(\xxx_2),\ldots,\varphi_\ell(\xxx_\ell)),
$$
is also an $(m,n,w)$-domination mapping for $m=\sum_{i=1}^\ell m_i$, $n=\sum_{i=1}^\ell n_i$,
and $w=\sum_{i=1}^\ell w_i$.
\end{theorem}

Domination mappings are not difficult to find (at least for small parameters). For example, in~\cite{CEKV18}\linebreak
$(2,3,1)$-domination mapping, $(9,15,3)$-domination mapping, and $(12,20,4)$-domination mapping were presented.
These three domination mappings have also efficient encoding and decoding procedures.

Certainly, one can use an $(m,n,w)$-domination mapping to form an $(n,t,w)$-LPC code
from an $(m,t)$ cooling code. The only question is whether there is an efficient encoding
and decoding schemes for the constructed LPC code. Such encoding and decoding schemes should be based on
efficient encoding and decoding schemes for both the related cooling code and the related
domination mapping. For large parameters such coding procedures might not exist.
Hence, it is better to use the product constructions using many domination mappings with small parameters,
but with efficient encoding and decoding schemes.
Our next construction for $(n,t,w)$-LPC codes is based on this idea. For demonstration we will use
a specific family of $(n,t,w)$-LPC code, but the same idea will work on any set of parameters that can be obtained
from domination mappings with small parameters by using Theorem~\ref{thm:gen_product}.
We will describe the construction via its encoding scheme.

\begin{construction}
\label{con:map}
Assume we are given $w \geq 6$, $m=3w=9\alpha+12\beta$,
$n=5w=15\alpha+20\beta$, $t$, and an $(m,t)$-cooling code $\mC$ with $2^k$ codesets.
We will construct an $(n,t,w)$-LPC code $\mC'$.
Let $\vu$ be the information word of length~$k$ and let $\mC_{\vu}$ be its
related codeset in $\mC$.
The encoder partitions the set of $m$ coordinates into $\alpha + \beta$ subsets,\linebreak
$\alpha$ subsets of size~9 and $\beta$ subsets of size~12. Similarly, it
partitions the set $n$ coordinates of the codewords from $\mC'$ into $\alpha + \beta$ subsets, $\alpha$
subsets of size~15 and $\beta$ subsets of size~20. Let $\varphi_1$ and $\varphi_2$
be a $(9,15,3)$-domination mapping and a $(12,20,4)$-domination mapping, respectively.
Let $\varphi$ be the $(m,n,w)$-domination mapping
implied by the product construction of Theorem~\ref{thm:gen_product} on $\alpha$ copies of $\varphi_1$ and $\beta$ copies
of $\varphi_2$.
Let $T$ be a $t$-subset of $[n]$ and
let $T' = N(T)$ be a $t'$-subset of $[m]$, where $t' \leq t$ by Lemma~\ref{lem:neighbour}.
The encoder finds the vector $\vv$ in~$\mC_{\vu}$ related to the set $T'$, i.e. $\vv$ has \emph{zeroes}
in the coordinates of $T'$, as required. Finally the encoder parse~$\vv$
into $\vv_1 \vv_2 \ldots \vv_\alpha \vv'_1 \vv'_2 \ldots \vv'_\beta$, where $\vv_i$ is of length 9 and $\vv'_i$ is
of length 12. By using the encodings of the mappings $\varphi_1$ and $\varphi_2$,
the encoder maps $\vv_1 \vv_2 \ldots \vv_\alpha \vv'_1 \vv'_2 \ldots \vv'_\beta$ to
the word $\varphi(\vv_1 \vv_2 \ldots \vv_\alpha \vv'_1 \vv'_2 \ldots \vv'_\beta)$
of the $(n,t,w)$-LPC code,
where each $\vv_i$ is mapped by $\varphi_1$ to a word of length 15 and each $\vv_i'$ is mapped by $\varphi_2$
to a word of length 20.

The decoder is applied in reverse order to generate the information word of length $k$
from a word of length~$n$ of the $(n=5w,t,w)$-LPC code $\mC'$, by first generating a word, of length $m$, from $\mC$
and after that using the decoder of $\mC$ to find the information word of length $m$.
\end{construction}

Note that Construction~\ref{con:map} can be viewed as a modification of the Concatenation
construction (See Proposition~\ref{prop:concat}).
Construction~\ref{con:map} has an advantage on the Concatenation of Proposition~\ref{prop:concat}
and other constructions with larger size for the same weight $w$ and the same number of hottest wires $t$.

How good are the codes constructed by using the domination mappings.
They are incomparable with the other codes which were constructed in previous sections
due to their parameters.
But, they can be easily compared with the codes obtained in Proposition~\ref{prop:concat}.
We will consider some examples by using three of the most simple (and less powerful) domination
mappings, a $(2,3,1)$-domination mapping, a $(3,9,15)$-domination mapping and a $(4,12,20)$-domination
mapping (note that the last two were used in Construction~\ref{con:map}).

To this end we will describe the simple and effective construction of cooling codes given in~\cite{Cheeetal.2017}.
This construction is based on spreads (or partial spreads) which will be defined next.

Loosely speaking, a partial $\tau$-spread of the vector
space $\F_q^n$ is a collection of disjoint
$\tau$-dimensional subspaces of~$\F_q^n$. Formally,
a collection $V_1,V_2,\ldots,V_M$ of $\tau$-dimensional
subspaces of $\F_q^n$ is said to be a \emph{partial $\tau$-spread of~$\F_q^n$} if
$$
V_i \cap V_j = \{\zero\}~
\text{for all\, $i \ne j$} ~,
$$
$$
\F_q^n \supseteq \kern1pt V_1 \cup V_2 \cup \cdots \cup V_M \hspace*{1ex}~.
$$
If the $\tau$-dimensional subspaces form a partition of $\F_q^n$ then
the partial $\tau$-spread is called a $\tau$-spread.
It is well known that such $\tau$-spreads exist if and
only~if~$\tau$~divides $n$, in which case
$M = (q^n{-}1)/(q^\tau{-}1) > q^{n-\tau}$.
For~the case where $\tau$ does not divide $n$,
\emph{partial $\tau$-spreads} with~\mbox{$M \ge q^{n-\tau}$}
have been constructed in~\cite[Theorem\,11]{EtVa11}.

\begin{theorem}
\label{thm:spread_cool}
Let ~ $V_1,\kern-1ptV_2,\ldots,V_M$ ~ be a partial ~ $(t{+}1)$-spread of~~$\F_2^n$,
and ~ define ~ the ~ code\kern1.5pt\ $\mC = \{V^*_1,V^*_2,\ldots,V^*_M\}$,
where
${V^*_i \kern-1.5pt=\kern-1pt V_i \kern1.5pt{\setminus} \{\zero\}}$
for all $i$.
Then $\mC$ is an $(n,t)$-cooling code of size $M \geq 2^{n-t-1}$.
\end{theorem}

Assume first that we want to construct an $(3w,t,w)$-LPC code from a $(2w,t)$-cooling code.
We consider the trivial $(2,3,1)$-domination mapping and use
a $(t+1)$-spread over $\F_2^{2w}$, where $2(t+1) \leq 2w$ to obtain a $(3w,t,w)$-LPC code
$\mC$ of size $\frac{2^{2w}-1}{2^{t+1}-1} > 2^{2w-t-1}$ (from a $(2w,t)$-cooling code) for any $t+1 \leq w$.
Assume now that we want to form a comparable code using Proposition~\ref{prop:concat}.
There are a few options that can be taken as parameters in Proposition~\ref{prop:concat}.
\begin{enumerate}
\item Assume first that we take $w' =1$, $s =3$, and $m=w$, in Proposition~\ref{prop:concat}.
As a consequence we can take $q=4$, and hence the size of the $(3w,t,w)$-LPC code
obtained by Proposition~\ref{prop:concat} will be $2^{2w-2t-1}$, for $t \leq 3$ and $t+1 \leq w/2$,
which is clearly much smaller
than $\mC$. Moreover, $t$ is at most the minimum between 3 and $\frac{w}{2}-1$
compared to $t \leq w-1$ for the code $\mC$ based on the $(2,3,1)$-domination mapping.

\item A different choice in Proposition~\ref{prop:concat} is $w'=3$, $s=9$, and $m=w/3$.
As a consequence we can take $q=128$, and hence the size of the $(3w,t,w)$-LPC code
obtained by Proposition~\ref{prop:concat} will be $2^{7w/3-7t-7}$ for any $t+1 \leq w/6$, but not larger than 9.
Hence, $t$ is much smaller compared to $t \leq w-1$ for the code $\mC$ based on the $(2,3,1)$-domination mapping.
As for the size of the code, for the same $t$ the code $\mC$ is larger if $t \geq \frac{w}{18} -1$
\end{enumerate}

We continue with our constructed $(5w,t,w)$-LPC code $\mC$ from a $(3w,t)$-cooling code
discussed in Construction~\ref{con:map}.
It has $2^{3w-t-1}$ codesets and it can handle any $t \leq 3w/2 -1$ (if the spread construction
of Theorem~\ref{thm:spread_cool} is used).
Assume now that we want to form a comparable code by using Proposition~\ref{prop:concat}.
There are a few options that can be taken as parameters in Proposition~\ref{prop:concat}.

\begin{enumerate}
\item Assume first that we take $w' =3$, $s =15$, and $m=w/3$, in Proposition~\ref{prop:concat}.
As a consequence we can take $q=2^9$, and hence the size of the $(5w,t,w)$-LPC code
obtained by Proposition~\ref{prop:concat} will be $2^{3w-9t-9}$, for any $t+1 \leq w/6$ but not larger than 15,
which is clearly much smaller than $\mC$ for both size and $t$.

\item A different choice in Proposition~\ref{prop:concat} is $w'=4$, $s=20$, and $m=w/4$.
As a consequence we can take $q=2^{12}$, and hence the size of the $(5w,t,w)$-LPC code
obtained by Proposition~\ref{prop:concat} will be $2^{3w-12t-12}$ for any $t+1 \leq w/8$, but not larger than 20,
which is clearly much smaller than $\mC$ for both size and $t$.
\end{enumerate}

It should be noted that $q$ can be sometimes slightly larger than the one given in the example.
This won't make much difference in the comparison, but the computation in a large field size which is not
a power of 2 is more messy.

We conclude that the codes obtained by this new method have in most cases larger size and
better capabilities than the best codes obtained by previous known constructions.

\vspace{-1ex}
\section*{Acknowledgments}
Y. M. Chee was supported in part by the Singapore Ministry of Education under grant MOE2017-T3-1-007.
T. Etzion and A. Vardy were supported in part by the BSF-NSF grant 2016692
Binational Science~Foundation (BSF), Jerusalem, Israel, under Grant 2012016.
The research of T. Etzion and A. Vardy was supported
in part by the National Science Foundation under grant CCF-1719139.
The research of Y. M. Chee, H. M. Kiah and H. Wei was supported in part  by the Singapore Ministry of Education
under grant MOE2015-T2-2-086.

\vspace{-1ex}

\bibliographystyle{plain}

\end{document}